\documentclass[11pt]{article}
\usepackage[margin=1in]{geometry}
\usepackage{amsmath,amssymb,amsthm,mathtools,bm}
\usepackage{newtxtext,newtxmath}

\usepackage{microtype}
\usepackage{booktabs,tabularx,array,longtable}
\usepackage{xcolor}
\usepackage[round,authoryear]{natbib}
\usepackage[colorlinks=true,allcolors=blue]{hyperref}
\usepackage{enumitem}
\usepackage{caption}

\setlist{leftmargin=1.6em,itemsep=0.2em,topsep=0.3em}

\newtheorem{theorem}{Theorem}[section]
\newtheorem{proposition}[theorem]{Proposition}
\newtheorem{lemma}[theorem]{Lemma}
\newtheorem{corollary}[theorem]{Corollary}
\newtheorem{assumption}[theorem]{Assumption}
\theoremstyle{remark}
\newtheorem{remark}[theorem]{Remark}
\newtheorem{example}[theorem]{Example}

\newcommand{\R}{\mathbb{R}}
\newcommand{\E}{\mathbb{E}}
\newcommand{\Var}{\operatorname{Var}}
\newcommand{\Cov}{\operatorname{Cov}}
\newcommand{\tr}{\operatorname{tr}}
\newcommand{\op}{\mathrm{op}}
\newcommand{\F}{\mathrm{F}}
\newcommand{\1}{\mathbf{1}}
\newcommand{\Ahat}{\widehat A}
\newcommand{\Abar}{A^0}
\newcommand{\Fhat}{\widehat F}
\newcommand{\neff}{n_{\mathrm{eff}}}
\newcommand{\rstab}{r_{\mathrm{stab}}}

\title{Spatial Kernel Restricted Canonical Correlation Analysis:\\Population Targets, Identifiability, and Stability}
\author{\begin{minipage}{0.97\textwidth}
\centering
Zhen Miao\textsuperscript{1}, Yilong Qu\textsuperscript{2}, and Nancy R. Zhang\textsuperscript{3}\\[0.6em]
\small
\textsuperscript{1}Fred Hutch Cancer Center, Seattle, WA, USA\\[0.3em]
\textsuperscript{2}Graduate Group in Genomics and Computational Biology,\\
Perelman School of Medicine, University of Pennsylvania, Philadelphia, PA, USA\\[0.3em]
\textsuperscript{3}Department of Statistics and Data Science, The Wharton School,\\
University of Pennsylvania, Philadelphia, PA, USA
\end{minipage}}
\date{}

\begin{document}
\maketitle

\begin{abstract}
Relating observations of different categories across space requires methods that accommodate distinct sampling locations and spatial dependence. Spatial kernel restricted canonical correlation analysis (skrCCA) uses spatial weights to identify associated feature combinations across data categories. Here, we develop a statistical theory of its covariance form after feature reduction. Conditional on the observed locations and under a mean-zero model, the pairwise association matrices target kernel-weighted averages of cross-covariances. For prespecified directions, we derive the exact null variance under independence between blocks. We establish recovery and stability bounds in two- and multiblock settings under explicit population models, and show that exact cyclic updates have convergent objective values. The variance characterization shows how spatial weighting and within-block dependence jointly determine the effective information available for estimating association. The recovery bounds separate statistical and optimization error, relating stability to population separation and the strength of connections between blocks. These results clarify what fitted spatial association axes estimate and provide conditions under which they can be recovered reliably.
\end{abstract}

\section{Introduction}
Many scientific studies collect multivariate observations from different populations or processes over a common spatial domain, each sampled at its own locations. A central statistical task is to identify associated patterns of variation across these data blocks. Spatial proximity provides a way to relate observations when direct pairing is unavailable. At the same time, nearby observations may carry partly redundant information. Understanding spatial association therefore requires both a population target and an account of its estimation uncertainty.

Spatial transcriptomics provides one motivating example. These technologies measure gene expression in its spatial context within tissue~\citep{MosesPachter2022}, allowing researchers to ask whether expression programs, coordinated patterns of variation across many genes, in different cell types vary together across space. Such coordination may accompany differentiation or tissue injury~\citep{Miao2026}. A cell at one position, however, does not have a unique counterpart of another type.

A similar problem arises in ecology when species composition and environmental conditions are recorded at different sets of sites within the same landscape~\citep{Dray2002}. One may seek combinations of species abundances associated with combinations of environmental variables across nearby sites. The data blocks have different features and sampling locations, but spatial proximity provides a basis for relating them. Both examples motivate a framework for studying association between spatially linked data blocks.

If observations were paired across blocks, canonical correlation analysis (CCA), including its multiblock extensions, could compare their feature combinations directly~\citep{Hotelling1936,Hardoon2004,Kettenring1971}. A coarse-grained remedy for missing pairing is to create common units by dividing the domain into bins and aggregating each block within them. CCA and sparse variants can then be applied to the matched profiles~\citep{WittenTibshirani2009}. In its spatial-niche formulation, DIALOGUE follows this route, averaging each cell type's features within niches before applying sparse CCA~\citep{JerbyArnon2022}. Aggregation discards positions within a bin, and two observations on opposite sides of the (arbitrary) boundary are treated as unrelated. The matched analysis also treats adjacent bins as independent units, discarding any higher-level structure.

Spatially weighted ordination instead lets a spatial weight matrix supply the pairing: the cross-covariance between a feature of one block and a feature of another is computed over all pairs of sites, each pair weighted by its proximity, so every observation stays at its own position. \citet{Wartenberg1985} introduced this device for one multivariate sample, replacing the covariance matrix by a spatially lagged cross-product and extracting its principal axes; \citet{Dray2008} generalized this approach to other one-table ordinations. For two tables observed at the same sites, co-inertia analysis~\citep{DoledecChessel1994} finds linear combinations, one per table, with maximal covariance. This is the singular value decomposition of the cross-covariance matrix, also known as inter-battery factor analysis, PLS--SVD, and maximum covariance analysis~\citep{Tucker1958,Wegelin2000,Bretherton1992}. \citet{Dray2002} extended co-inertia analysis to two samples at different sites by linking the sites of one sample to nearby sites of the other through a neighborhood graph; the result is an instance of RLQ analysis~\citep{Doledec1996}, in which two tables are linked through a third.

CoPro follows this spatially weighted route and introduces spatial kernel restricted canonical correlation analysis (skrCCA), which connects cells of different types through a spatial kernel~\citep{Miao2026}. With a Gaussian kernel, the bandwidth specifies the spatial scale of the association being summarized. Such criteria define how spatial weights determine the association axes. Relating those axes to population associations and quantifying their sampling variability requires an explicit statistical model.

We study skrCCA as a general framework for associating spatial data blocks observed at different locations. The organizing question is when spatially coupled, unpaired data support identifiable and reliably estimated association axes.

We analyze the covariance form after feature reduction, in the regime where the number of retained features or principal components is small relative to the number of spatial units. Conditional on the observed positions, we identify the population targets and derive sampling-variability bounds under explicit covariance models. We connect matrix-estimation error to recovery of association directions under spectral-separation or structured-factor assumptions, separate statistical error from optimization error, and quantify how earlier-axis error propagates through sequential deflation. Together, these results describe how spatial dependence, numerical approximation, and sequential extraction affect the interpretation of fitted associations.

The paper is organized as follows. Section~\ref{sec:operators} defines the association matrices and their population targets. Section~\ref{sec:null} derives null laws and matrix-error bounds, and Section~\ref{sec:two-block} studies recovery for two blocks. Sections~\ref{sec:optimization} and~\ref{sec:multiblock-identifiability} address computation and population recovery for several blocks, respectively. Section~\ref{sec:deflation} treats higher-order components and deflation. Section~\ref{sec:diagnostics} translates the results into practitioner diagnostics, and Section~\ref{sec:discussion} discusses their implications and limitations.

\section{Kernel-coupled association matrices}
\label{sec:operators}
To study spatial association across types, a criterion must be specified that links observations of different types. This section examines how the components of that criterion determine the association being estimated. A fitted direction can then be read as an estimate of a population quantity, and the analysis choices that change that quantity can be separated from those that leave it unchanged.

\subsection{Data blocks, spatial kernels, and pair scaling}
\label{sec:blocks}
We use a block to mean the observations belonging to one biological type or data modality. In the spatial-transcriptomics application of CoPro, block $b$ contains the cells assigned to cell type $b$; its rows are cells and its columns are gene-level features or retained principal components. More generally, a block may represent any collection of spatial units for which the same set of features is measured.

Let $B\ge2$ denote the number of blocks. Block $b$ contains $n_b$ spatial units and $p_b$ features,
\begin{equation}
X_b=
\begin{pmatrix}
x_b(s^{(b)}_1)^\top\\
\vdots\\
x_b(s^{(b)}_{n_b})^\top
\end{pmatrix}\in\R^{n_b\times p_b},\qquad b=1,\ldots,B.
\label{eq:data}
\end{equation}
For the algebraic and optimization results, empirical centering may be imposed as a preprocessing convention,
\begin{equation}
\1_{n_b}^\top X_b=0\in\R^{1\times p_b}.
\label{eq:centering}
\end{equation}
The stochastic results of Section~\ref{sec:null} instead place a probability model directly on the score vectors or data matrices.

For two blocks $b,c$, let
\[
K_{bc}\in\R^{n_b\times n_c},\qquad K_{cb}=K_{bc}^\top,
\]
be a deterministic nonnegative kernel matrix. The kernel need not be Gaussian; we assume only positive total mass
\begin{equation}
\kappa_{bc}=\1_{n_b}^\top K_{bc}\1_{n_c}>0.
\label{eq:kappa}
\end{equation}
Dividing by total kernel mass expresses each pairwise association as a weighted average. Define the mass-normalized pair association matrix
\begin{equation}
\Ahat_{bc}=\frac{1}{\kappa_{bc}}X_b^\top K_{bc}X_c\in\R^{p_b\times p_c},\qquad \Ahat_{cb}=\Ahat_{bc}^\top.
\label{eq:Ahat}
\end{equation}
The pair matrix \eqref{eq:Ahat} is the cross-type analogue, with a rectangular kernel, of Wartenberg's spatially lagged cross-product $Z^\top WZ$ for a single sample~\citep{Wartenberg1985}. For nonnegative symmetric pair weights $w_{bc}=w_{cb}$ with $w_{bb}=0$, define
\begin{equation}
\Fhat(u_1,\ldots,u_B)=\sum_{1\le b<c\le B}w_{bc}u_b^\top\Ahat_{bc}u_c,
\qquad u_b\in\mathbb S^{p_b-1}.
\label{eq:multiobjective}
\end{equation}
This is a SUMCOR-type multiset association criterion~\citep{Kettenring1971}. With Euclidean unit constraints, it is a weighted MAXDIFF criterion with spatial kernels replacing the usual cross-products~\citep{tenBergeKnol1984,vandeGeer1984,tenBerge1988,HanafiKiers2006}. Using the unnormalized matrix $X_b^\top K_{bc}X_c$ is equivalent to replacing $w_{bc}$ by $w_{bc}\kappa_{bc}$. With two blocks this positive scalar does not change singular vectors. With three or more blocks, pair scaling changes relative edge weights and can change fitted directions. We allow zero pair weights to cover analyses restricted to selected pairs.

For a single tissue section, CoPro (package version 1.3.0) maximizes $\sum_{b<c}u_b^\top X_b^\top K_{bc}X_cu_c$ over unit vectors, which is \eqref{eq:multiobjective} with $w_{bc}=\kappa_{bc}$. When $S$ sections are fitted jointly with shared directions under its summed-covariance option, the objective is $\sum_s\sum_{b<c}u_b^\top X_{b,s}^\top K_{bc,s}X_{c,s}u_c$, which is again \eqref{eq:multiobjective} with
\begin{equation}
\Ahat_{bc}=\frac{\sum_{s}X_{b,s}^\top K_{bc,s}X_{c,s}}{\kappa_{bc,\rm tot}},\qquad
\kappa_{bc,\rm tot}=\sum_{s}\kappa_{bc,s},\qquad w_{bc}=\kappa_{bc,\rm tot}.
\label{eq:multisection}
\end{equation}
The results below apply to \eqref{eq:multisection} with the block-diagonal kernel over the stacked sections in place of a single-section kernel. CoPro's default multisection objective instead sums section-specific correlation ratios with direction-dependent denominators and lies outside the present analysis (Section~\ref{sec:discussion}). For the covariance-form problem, Section~\ref{sec:two-block-solution} gives the exact two-block SVD solution; Sections~\ref{sec:optimization} and~\ref{sec:deflation} analyze exact cyclic updates and projection deflation for multiple blocks.

\subsection{Finite-design population target}
We treat the observed spatial positions as fixed. Let
\[
C_{bc}(s,t)=\E\{x_b(s)x_c(t)^\top\}\in\R^{p_b\times p_c}
\]
be the cross-second-moment array of the input features, conditional on the full observed design. If the features are empirically centered, the $n_b\times n_c$ matrix of raw cross-moments for each feature pair is pre- and post-multiplied by the centering projections $H_b=I_{n_b}-n_b^{-1}\1\1^\top$ and $H_c$, so the centered array $C_{bc}$ depends on the full set of observed sites. The expectation of the pair matrix is
\begin{equation}
\Abar_{bc}:=\E(\Ahat_{bc})
=\frac{1}{\kappa_{bc}}\sum_{i=1}^{n_b}\sum_{j=1}^{n_c}(K_{bc})_{ij}
C_{bc}(s_i^{(b)},s_j^{(c)}).
\label{eq:pop}
\end{equation}
Under a mean-zero model this is a kernel-weighted average of cross-covariances at the observed spatial separations. Changing kernel bandwidth or shape generally changes $\Abar_{bc}$ and therefore changes the population singular system, so sensitivity to the kernel can reflect a change in the spatial scale being summarized.

\subsection{Two-block solution}
\label{sec:two-block-solution}
For $B=2$, write $\Ahat=\Ahat_{12}\in\R^{p\times q}$ with $p=p_1$ and $q=p_2$, and write $u=u_1$, $v=u_2$. The order-one problem is
\begin{equation}
\max_{\|u\|_2=\|v\|_2=1}u^\top\Ahat v.
\label{eq:two}
\end{equation}
By the variational characterization of singular values, its optimum is the leading singular value $\widehat\sigma_1$ of $\Ahat$, attained at the leading left and right singular vectors up to a joint sign. Successive within-block orthogonality constraints give the $k$th singular pair for $k\le\min(p,q)$, and within a tied singular subspace the basis is not unique.

We retain CoPro's name skrCCA for this covariance criterion, whose value need not be bounded by one; Section~\ref{sec:standardization} treats the standardized forms.

\subsection{Normalization and preprocessing invariances}
Two elementary invariances fix the baseline for interpreting sensitivity to sampling and preprocessing. First, hold the pair weights $w_{bc}$ in \eqref{eq:multiobjective} fixed and replicate every observation of block $b$, together with its location, exactly $m_b\ge1$ times, repeating the corresponding kernel rows and columns. The numerator and the denominator of \eqref{eq:Ahat} are both multiplied by $m_bm_c$, so every mass-normalized pair matrix and the objective are unchanged. This concerns deterministic multiplicity only: a smaller or spatially less diverse block can still have larger variance, and selective spatial sampling changes the weighted design in \eqref{eq:pop}. If the unnormalized objective is used, with implicit weights proportional to $\kappa_{bc}$, uniform replication changes those weights and the invariance no longer holds.

Second, let $P_b\in\R^{p_b\times r_b}$ have orthonormal columns, let $Z_b=X_bP_b$ be a reduced feature representation, and let $\Ahat^P_{bc}=P_b^\top\Ahat_{bc}P_c$ be the reduced pair matrix. Replacing $P_b$ by $P_bQ_b$ with $Q_b$ orthogonal replaces $\Ahat^P_{bc}$ by $Q_b^\top\Ahat^P_{bc}Q_c$ and a solution $a_b$ by $Q_b^\top a_b$, while the reconstructed scores $X_bP_ba_b$ are unchanged. A program interpreted through its cell scores therefore has the same meaning in every orthonormal basis of the retained subspace.

\subsection{Standardized spatial association and whitening}
\label{sec:standardization}
Feature scaling determines which combinations a covariance objective favors. Fixed marginal metrics can be absorbed by a change of coordinates, whereas standardization using the spatial kernel's marginal weights defines a different target.

\begin{proposition}[Fixed-metric whitening equivalence]
\label{prop:fixedmetric}
Let $G_x\in\R^{p\times p}$ and $G_y\in\R^{q\times q}$ be fixed positive-definite matrices. Consider
\begin{equation}
\max_{u^\top G_xu=1,\;v^\top G_yv=1}u^\top\Ahat v.
\label{eq:fixedmetric}
\end{equation}
Define
\begin{equation}
\widetilde A=G_x^{-1/2}\Ahat G_y^{-1/2}.
\label{eq:whitenedA}
\end{equation}
Then \eqref{eq:fixedmetric} is exactly equivalent, under $\widetilde u=G_x^{1/2}u$ and $\widetilde v=G_y^{1/2}v$, to
\[
\max_{\|\widetilde u\|_2=\|\widetilde v\|_2=1}\widetilde u^\top\widetilde A\widetilde v.
\]
Consequently the maximizing directions are obtained from the singular vectors of $\widetilde A$ by
$u=G_x^{-1/2}\widetilde u$ and $v=G_y^{-1/2}\widetilde v$. The same coordinate change applies blockwise to the multiblock objective when each block uses one fixed positive-definite metric.
\end{proposition}
\begin{proof}
Substitute $u=G_x^{-1/2}\widetilde u$ and $v=G_y^{-1/2}\widetilde v$; both maps are bijective because the metrics are positive definite.
\end{proof}

A common score-standardized two-block criterion takes
\begin{equation}
\mathcal C_{\rm marg}(u,v)
=\frac{u^\top\Ahat v}
{\{u^\top S_xu\}^{1/2}\{v^\top S_yv\}^{1/2}},
\qquad
S_x=\frac{1}{n_x}X^\top X,\quad S_y=\frac{1}{n_y}Y^\top Y.
\label{eq:margstandard}
\end{equation}
For observed positive-definite $S_x,S_y$, Proposition~\ref{prop:fixedmetric} applies conditionally on those matrices. Thus the optimization is an SVD after marginal whitening. The value in \eqref{eq:margstandard}, however, need not be bounded by one, because its cross-term and its two marginal quadratic forms need not be blocks of the same positive-semidefinite covariance matrix.

\begin{corollary}[Exact coincidence after block whitening]
\label{cor:pca-whitening}
Suppose preprocessing yields $S_x=c_xI_p$ and $S_y=c_yI_q$ for constants $c_x,c_y>0$. Then, over Euclidean unit vectors,
\[
\mathcal C_{\rm marg}(u,v)=\frac{1}{\sqrt{c_xc_y}}u^\top\Ahat v.
\]
Hence the score-standardized criterion \eqref{eq:margstandard} and the covariance-form criterion \eqref{eq:two} have exactly the same maximizing directions and singular subspaces. In particular, if retained PCA scores are scaled so that their empirical covariance is the identity, then $c_x=c_y=1$.
\end{corollary}
\begin{proof}
For unit $u,v$, $u^\top S_xu=c_x$ and $v^\top S_yv=c_y$, so the denominator in \eqref{eq:margstandard} is constant over the feasible set.
\end{proof}

Replacing the marginal covariances by the identity is also the device used in penalized matrix-decomposition CCA~\citep{Witten2009}.

The criterion \eqref{eq:margstandard} also differs from CoPro's reported kernel-normalized association score~\citep{Miao2026}. For centered scores $a=Xu$ and $b=Yv$, the implemented statistic (package version 1.3.0) is
\begin{equation}
\mathrm{NCorr}(u,v)=\frac{a^\top Kb}{\|a\|_2\,\|b\|_2\,\|R_x^{1/2}K_cR_y^{1/2}\|_\F},
\qquad K_c=H_xKH_y,
\label{eq:ncorr}
\end{equation}
where $H_x,H_y$ are the centering projections and $R_x,R_y$ are within-type correlation operators, treated as fixed in the variance calculation below. By Theorem~\ref{thm:variance}, if $a$ and $b$ are independent with covariances $\sigma_a^2H_xR_xH_x$ and $\sigma_b^2H_yR_yH_y$, the kernel factor in \eqref{eq:ncorr} equals the null standard deviation of $a^\top Kb$ divided by $\sigma_a\sigma_b$, so \eqref{eq:ncorr} is the null-standardized numerator divided by $\|a\|_2\|b\|_2/(\sigma_a\sigma_b)$. The score norms are random and computed from the same data, and the directions are fitted. Thus \eqref{eq:ncorr} is a descriptive scale for comparing bandwidths rather than an exactly standardized statistic, and the fixed-direction results do not supply its null law for fitted directions.

A correlation coefficient requires the cross-term and both marginal variances to refer to the same distribution of paired observations. A spatial kernel supplies such a distribution by treating its nonnegative weights as masses on cross-type edges.

For two blocks, define
\[
d_x=K\1_{n_y},\qquad d_y=K^\top\1_{n_x},\qquad
D_x=\operatorname{diag}(d_x),\qquad D_y=\operatorname{diag}(d_y),\qquad
\kappa=\1_{n_x}^\top K\1_{n_y}.
\]
Define the kernel-edge marginal means
\begin{equation}
\mu_{x,K}=\frac{1}{\kappa}X^\top d_x,
\qquad
\mu_{y,K}=\frac{1}{\kappa}Y^\top d_y,
\label{eq:edge-means}
\end{equation}
and centered matrices
$X_K=X-\1_{n_x}\mu_{x,K}^\top$ and $Y_K=Y-\1_{n_y}\mu_{y,K}^\top$. Define
\begin{equation}
\Sigma_{xx,K}=\frac{1}{\kappa}X_K^\top D_xX_K,\quad
\Sigma_{yy,K}=\frac{1}{\kappa}Y_K^\top D_yY_K,\quad
\Sigma_{xy,K}=\frac{1}{\kappa}X_K^\top K Y_K.
\label{eq:edge-covs}
\end{equation}

\begin{proposition}[Kernel-edge canonical correlation]
\label{prop:edgecca}
Let a random pair of row indices $(I,J)$ have probability
\begin{equation}
\Pr(I=i,J=j)=\frac{K_{ij}}{\kappa}.
\label{eq:edge-dist}
\end{equation}
Then the three matrices in \eqref{eq:edge-covs} are exactly the covariance blocks of the random endpoint pair $(x_I,y_J)$ after subtracting the means in \eqref{eq:edge-means}. Consequently,
\begin{equation}
\begin{pmatrix}
\Sigma_{xx,K}&\Sigma_{xy,K}\\
\Sigma_{xy,K}^\top&\Sigma_{yy,K}
\end{pmatrix}\succeq0.
\label{eq:edge-psd}
\end{equation}
If $\Sigma_{xx,K}$ and $\Sigma_{yy,K}$ are positive definite, define
\begin{equation}
R_K=\Sigma_{xx,K}^{-1/2}\Sigma_{xy,K}\Sigma_{yy,K}^{-1/2}.
\label{eq:edge-R}
\end{equation}
Every singular value of $R_K$ lies in $[0,1]$, and the leading singular value equals
\begin{equation}
\max_{u^\top\Sigma_{xx,K}u=1,\;v^\top\Sigma_{yy,K}v=1}
 u^\top\Sigma_{xy,K}v.
\label{eq:edge-cca}
\end{equation}
Thus \eqref{eq:edge-cca} is a canonical correlation for a cross-type pair drawn with probability proportional to the spatial kernel weight.
\end{proposition}
\begin{proof}
Under \eqref{eq:edge-dist}, $\Pr(I=i)=(d_x)_i/\kappa$ and $\Pr(J=j)=(d_y)_j/\kappa$. Therefore the expectations, marginal covariance matrices, and cross-covariance of $(x_I,y_J)$ are exactly \eqref{eq:edge-means}--\eqref{eq:edge-covs}; hence \eqref{eq:edge-psd} is a covariance matrix and is positive semidefinite. Proposition~\ref{prop:fixedmetric} gives the variational formula \eqref{eq:edge-cca}. For unit $a,b$, the random variables $a^\top\Sigma_{xx,K}^{-1/2}(x_I-\mu_{x,K})$ and $b^\top\Sigma_{yy,K}^{-1/2}(y_J-\mu_{y,K})$ each have variance one. Cauchy--Schwarz therefore gives $|a^\top R_Kb|\le1$. Taking the supremum over unit $a,b$ gives $\|R_K\|_{\op}\le1$, and all singular values are nonnegative and no larger than one.
\end{proof}

Proposition~\ref{prop:edgecca} is the kernel analogue of RLQ analysis, in which the link table is treated by correspondence analysis~\citep{Doledec1996,Dray2002}. If the data have first been centered using the kernel-edge means, then $\Sigma_{xy,K}$ equals the mass-normalized cross-term in \eqref{eq:Ahat}; the difference is that the marginal metrics in \eqref{eq:edge-covs} use the same edge distribution as that cross-term. Ordinary PCA whitening with respect to $X^\top X/n_x$ and $Y^\top Y/n_y$ does not generally make $\Sigma_{xx,K}$ and $\Sigma_{yy,K}$ equal to identity, because the kernel-edge marginal weights need not be uniform.

\section{Statistical variability of pairwise association matrices}
\label{sec:null}
Spatially smooth expression patterns can yield appreciable sample cross-association even when the two types are independent. Assessing that association requires a variability calculation that accounts for dependence within each type and the reuse of observations across kernel-weighted pairs. Recovery of fitted programs further requires a measure of error for the entire association matrix.

\subsection{Mean and variance under cross-block independence}
For prespecified gene programs, the first calibration question is how much their weighted association would fluctuate under cross-type independence.

Consider $X\in\R^{n_x\times p}$ and $Y\in\R^{n_y\times q}$ and fixed $u\in\R^p$, $v\in\R^q$. Define score vectors
\begin{equation}
a=Xu\in\R^{n_x},\qquad b=Yv\in\R^{n_y},
\label{eq:scores}
\end{equation}
and
\begin{equation}
T(u,v)=u^\top\Ahat v=\frac{a^\top Kb}{\kappa},\qquad \kappa=\1^\top K\1.
\label{eq:T}
\end{equation}
The statistic $T(u,v)$ is a Mantel-type weighted cross-product~\citep{Mantel1967,Wartenberg1985}. For the null calculation, $a,b$ are random while $u,v,K$ and the positions are fixed. If the scores are empirically centered, the centering projection is part of the map generating them: with $H_x=I-n_x^{-1}\1\1^\top$ and $H_y=I-n_y^{-1}\1\1^\top$, the centered scores $a=H_xa_{\rm raw}$ and $b=H_yb_{\rm raw}$ have covariances $H_x\Gamma_{a,\rm raw}H_x$ and $H_y\Gamma_{b,\rm raw}H_y$, and the covariances below are those of the scores that enter the statistic. When PCA loadings or whitening matrices are learned from the same observations, the distributional assumptions must hold for the transformed inputs; Gaussian raw data alone do not establish them.

\begin{theorem}[Exact fixed-direction null variance]
\label{thm:variance}
Assume $a$ and $b$ are independent with $\E a=0$ and $\E b=0$. Let $\Gamma_a=\Cov(a)$ and $\Gamma_b=\Cov(b)$. Then
\begin{equation}
\E_0T(u,v)=0,\qquad
\Var_0\{T(u,v)\}
=\frac{\tr(\Gamma_aK\Gamma_bK^\top)}{\kappa^2}
=\frac{\|\Gamma_a^{1/2}K\Gamma_b^{1/2}\|_\F^2}{\kappa^2}.
\label{eq:variance}
\end{equation}
\end{theorem}
\begin{proof}
Independence and zero means give $\E(a^\top Kb)=0$. Moreover,
\[
\E(a^\top Kb)^2
=\E\{a^\top Kbb^\top K^\top a\}
=\E\{a^\top K\Gamma_bK^\top a\}
=\tr(\Gamma_aK\Gamma_bK^\top).
\]
The Frobenius identity follows by cyclicity of trace. Divide by $\kappa^2$.
\end{proof}
The same trace form appears, for a square weight matrix on a common set of sites, in the sample-covariance variance of \citet[eq.~2.2]{Dutilleul1993}, derived there under joint Gaussianity. Equation~\eqref{eq:variance} needs only independence and finite second moments, and with identity pairing it reduces to the variance used by \citet[eq.~2.3]{Clifford1989}.

To express~\eqref{eq:variance} on a familiar sample-size scale, we use the average marginal variances $\sigma_a^2=\tr(\Gamma_a)/n_x$ and $\sigma_b^2=\tr(\Gamma_b)/n_y$. Define, when the null variance is positive,
\begin{equation}
\neff(u,v)=\frac{\sigma_a^2\sigma_b^2}{\Var_0\{T(u,v)\}}
=\frac{\kappa^2\sigma_a^2\sigma_b^2}{\tr(\Gamma_aK\Gamma_bK^\top)}.
\label{eq:neff}
\end{equation}
Equivalently, put $R_a=\Gamma_a/\sigma_a^2$ and $R_b=\Gamma_b/\sigma_b^2$ to obtain $\neff=\kappa^2/\|R_a^{1/2}KR_b^{1/2}\|_\F^2$. The matrices $R_a,R_b$ need not have unit diagonals after empirical centering; only under constant position variance are they correlation matrices.

For uncentered scores that are homoscedastic and uncorrelated across positions, the average-variance convention gives $R_a=R_b=I$; if also $n_x=n_y=n$ and $K=I_n$, then $\neff=n$, and for general $K$, $\neff=\kappa^2/\|K\|_\F^2$. With identity pairing and homoscedastic spatial autocorrelation, $\neff=n^2/\tr(R_aR_b)$, matching the variance calculation of \citet[eqs.~2.6--2.10]{Clifford1989}. Analytic bivariate-Moran nulls that treat the spatial units as independent, as in SpatialDM~\citep{Li2023}, correspond to $R_a=R_b=I$. For centered sample correlations, \citet[Sections~2--3]{Dutilleul1993} incorporates the centering projection in the covariance calculation and proposes an approximate $t$ test with estimated effective degrees of freedom. Here $\neff$ summarizes the variance of a bilinear numerator; it does not by itself supply a $t$ reference distribution.

For the general formula, $R_a,R_b$ may depend on the chosen directions. A separable row--feature covariance model~\citep{Dawid1981} makes them direction-independent. Let $R_x,R_y$ be row-correlation matrices, $\Sigma_x,\Sigma_y$ positive-semidefinite feature covariance matrices, and $Z_x,Z_y$ have independent mean-zero unit-variance entries. Under
\begin{equation}
X=R_x^{1/2}Z_x\Sigma_x^{1/2},\qquad
Y=R_y^{1/2}Z_y\Sigma_y^{1/2},
\label{eq:separable}
\end{equation}
we have $\Cov(Xu)=(u^\top\Sigma_xu)R_x$ and similarly for $Yv$ (after empirical centering, read $H_xR_xH_x$ for $R_x$ and $H_yR_yH_y$ for $R_y$).

\begin{remark}[Fixed standardization versus a random denominator]
Theorem~\ref{thm:variance} applies unchanged after any \emph{fixed} marginal reparameterization from Proposition~\ref{prop:fixedmetric}: a fixed direction in whitened coordinates corresponds to a fixed transformed direction in the original coordinates, and the theorem calibrates its bilinear numerator exactly. A statistic that divides by marginal scales estimated from the same data has a random denominator, dependent on the numerator, and is not covered; Section~\ref{sec:whitening-perturbation} treats the estimated metrics as a separate perturbation.
\end{remark}

\subsection{Exact Gaussian representation and stable-rank normal approximation}
Under Gaussian score models, the kernel and the spatial covariances determine a spectral representation of the statistic, and the relative contributions of its singular modes determine the accuracy of a normal approximation.

Assume in addition that $a=\sigma_aR_a^{1/2}g$ and $b=\sigma_bR_b^{1/2}g'$ with independent standard Gaussian vectors. Let
\begin{equation}
M=R_a^{1/2}KR_b^{1/2},
\label{eq:M}
\end{equation}
and let $d_1,\ldots,d_r$ be its nonzero singular values.
\begin{theorem}[Exact Gaussian spectral representation]
\label{thm:gausslaw}
Under the preceding assumptions,
\begin{equation}
\frac{\kappa T(u,v)}{\sigma_a\sigma_b}
\overset d=\sum_{\ell=1}^r d_\ell\xi_\ell\eta_\ell,
\label{eq:gausslaw}
\end{equation}
where $\xi_\ell,\eta_\ell$ are mutually independent standard normal variables.
\end{theorem}
\begin{proof}
Write an SVD $M=UDV^\top$. Rotational invariance gives $U^\top g$ and $V^\top g'$ as independent standard Gaussian vectors, and $g^\top Mg'=\sum_\ell d_\ell\xi_\ell\eta_\ell$.
\end{proof}

The spectral representation suggests measuring how broadly the null variance is distributed across independent modes. Define stable rank $\rstab(M)=\|M\|_\F^2/\|M\|_\op^2$ for $M\ne0$.
\begin{corollary}[Berry--Esseen bound]
Assume $M\ne0$. There is a universal constant $C>0$ such that
\begin{equation}
\sup_z\left|
\Pr\left\{\frac{\kappa T(u,v)}{\sigma_a\sigma_b\|M\|_\F}\le z\right\}-\Phi(z)
\right|
\le C\frac{\sum_\ell|d_\ell|^3}{(\sum_\ell d_\ell^2)^{3/2}}
\le \frac{C}{\sqrt{\rstab(M)}}.
\label{eq:BE}
\end{equation}
\end{corollary}
\begin{proof}
The summands $d_\ell\xi_\ell\eta_\ell$ are independent, centered, and have variance $d_\ell^2$. Classical Berry--Esseen theory~\citep{Petrov1975} gives the first bound because $\E|\xi\eta|^3<\infty$. The second follows from $\sum|d_\ell|^3\le(\max|d_\ell|)\sum d_\ell^2$.
\end{proof}

If $M=0$, Theorem~\ref{thm:gausslaw} gives $T(u,v)=0$ almost surely; this occurs with positive kernel mass for empirically centered scores and $K=\1_{n_x}\1_{n_y}^\top$.

\subsection{From fixed-direction tails to matrix-norm concentration}
Fitted directions are selected by searching over feature combinations, so their analysis requires simultaneous control across directions. An operator-norm bound provides this control by bounding the largest bilinear fluctuation of the association matrix.

\begin{lemma}[Two-sided Gaussian quadratic-form tail]
\label{lem:gaussian-tail}
For symmetric $W$, $g\sim N(0,I)$, and $t>0$,
\[
\Pr\{|g^\top Wg-\tr W|>2\|W\|_\F\sqrt t+2\|W\|_\op t\}
\le2e^{-t}.
\]
\end{lemma}
\begin{proof}
Diagonalize $W$ with eigenvalues $\lambda_j$. For $2|z|\|W\|_\op<1$, the power series for the logarithm gives
\[
\log\E e^{z\lambda_j(g_j^2-1)}
=-z\lambda_j-\tfrac12\log(1-2z\lambda_j)
\le\frac{z^2\lambda_j^2}{1-2|z||\lambda_j|}.
\]
Summing bounds the log MGF by $vz^2/\{2(1-c|z|)\}$, where $v=2\|W\|_\F^2$ and $c=2\|W\|_\op$. For $W\ne0$, Chernoff's inequality with $z=\sqrt{2t/v}/(1+c\sqrt{2t/v})$ bounds each tail at $\sqrt{2vt}+ct$ by $e^{-t}$. Apply the same argument to $-W$ and sum. The case $W=0$ is immediate. The bound is the two-sided, signed-eigenvalue form of the chi-square tail inequality of \citet{LaurentMassart2000}, and the Chernoff step is the sub-gamma bound of \citet{BLM2013}; for sub-Gaussian analogues see the Hanson--Wright inequality~\citep{RudelsonVershynin2013}.
\end{proof}

\begin{theorem}[Matrix concentration under a separable Gaussian null]
\label{thm:operator-concentration}
Assume \eqref{eq:separable}, with $Z_x,Z_y$ independent standard Gaussian matrices. Let $M=R_x^{1/2}KR_y^{1/2}\ne0$, $N_K=\kappa^2/\|M\|_\F^2$, $\rstab=\|M\|_\F^2/\|M\|_\op^2$, and $h=t+(p+q)\log9$. Under the cross-block null, $\Abar=0$. For every $t>0$, with probability at least $1-2e^{-t}$,
\begin{equation}
\|\Ahat-\Abar\|_\op
\le \|\Sigma_x\|_\op^{1/2}\|\Sigma_y\|_\op^{1/2}
\left\{
2\sqrt{\frac{2h}{N_K}}+
\frac{2h}{\sqrt{N_K\rstab}}
\right\}.
\label{eq:operatorconc}
\end{equation}
\end{theorem}
\begin{proof}
For independent standard normal $g,g'$, write $g^\top Mg'$ as a quadratic form with symmetric block matrix $W=\tfrac12\left(\begin{smallmatrix}0&M\\M^\top&0\end{smallmatrix}\right)$. Its trace is zero and its norms are $\|W\|_\F=\|M\|_\F/\sqrt2$ and $\|W\|_\op=\|M\|_\op/2$. Lemma~\ref{lem:gaussian-tail} gives
\[
\Pr\{|g^\top Mg'|>\sqrt2\|M\|_\F\sqrt t+\|M\|_\op t\}\le2e^{-t}.
\]
Apply this to each fixed unit pair $u,v$, using its feature-scale factor at most $\|\Sigma_x\|_\op^{1/2}\|\Sigma_y\|_\op^{1/2}/\kappa$. Two $1/4$-nets have at most $9^{p+q}$ pairs; a union bound at $h$ and the factor-two net inequality~\citep{Vershynin2018} give \eqref{eq:operatorconc}.
\end{proof}

Here $N_K$ is the effective sample size implied by the row factors in~\eqref{eq:separable}. It equals the fixed-direction $\neff$ when the declared variance scales are $u^\top\Sigma_xu$ and $v^\top\Sigma_yv$; it need not equal the trace-normalized summary after centering.

\subsection{Exact matrix-noise scale, including nonnull association}
\label{sec:matrix-noise}
The null calculations above characterize spurious association. Recovery of a nonzero population component instead requires control of $\Ahat-\Abar$ when the two blocks are genuinely associated. Under joint Gaussianity, the next result gives the corresponding exact entrywise second moments.
\begin{proposition}[Nonnull Gaussian entry variance]
\label{prop:nonnull-variance}
For jointly mean-zero Gaussian $a,b$, let $\Gamma_a=\Cov(a)$, $\Gamma_b=\Cov(b)$, $C=\Cov(a,b)$, and let $L$ be fixed. Then
\[
\E(a^\top Lb)=\tr(LC^\top),\qquad
\Var(a^\top Lb)=\tr(\Gamma_aL\Gamma_bL^\top)+\tr\{(LC^\top)^2\}.
\]
\end{proposition}
\begin{proof}
Let $\Omega=\left(\begin{smallmatrix}\Gamma_a&C\\C^\top&\Gamma_b\end{smallmatrix}\right)$ and $Q=\tfrac12\left(\begin{smallmatrix}0&L\\L^\top&0\end{smallmatrix}\right)$. Writing $(a,b)=\Omega^{1/2}g$ gives a Gaussian quadratic form with variance $2\tr\{(\Omega^{1/2}Q\Omega^{1/2})^2\}$. Expand its block products. The cross-covariance trace need not be nonnegative by itself; only the complete variance must be.
\end{proof}

Summing the entrywise variances converts these calculations into a single measure of matrix error that can be compared with the strength and separation of population components.

\begin{corollary}[Exact second moment and an operator-error bound]
\label{cor:matrix-rms}
For any random $p$-by-$q$ matrix $\Ahat$ with finite entrywise variances and $\Abar=\E\Ahat$, define
\begin{equation}
V_A=\sum_{i=1}^p\sum_{j=1}^q\Var(\Ahat_{ij})
=\E\|\Ahat-\Abar\|_\F^2,\qquad s_A=\sqrt{V_A}.
\label{eq:matrix-rms}
\end{equation}
For $0<\alpha<1$,
\begin{equation}
\Pr\{\|\Ahat-\Abar\|_\op>\epsilon_{A,\alpha}\}\le\alpha,
\qquad \epsilon_{A,\alpha}=\sqrt{V_A/\alpha}.
\label{eq:matrix-markov}
\end{equation}
For $\Ahat=X^\top LY$, the entry variances are given by Theorem~\ref{thm:variance} under mean-zero block independence, or Proposition~\ref{prop:nonnull-variance} under joint mean-zero Gaussianity, using $a=Xe_i,b=Ye_j$. In this paper $L=K/\kappa$.
\end{corollary}
\begin{proof}
Expand the squared Frobenius norm entrywise; no independence between entries is needed. Apply Markov's inequality to $\|\Ahat-\Abar\|_\F^2$ and use $\|E\|_\op\le\|E\|_\F$.
\end{proof}
Under the separable null, $V_A=\tr(\Sigma_x)\tr(\Sigma_y)\|M\|_\F^2/\kappa^2$. The second-moment scale $s_A$ avoids the constants introduced by the covering argument and can be useful with few retained features. The 95\% Markov radius is $s_A/\sqrt{0.05}$; neither it nor the net bound dominates the other at matched confidence levels.

The covariance quantities entering $V_A$ are model quantities. Equation~\eqref{eq:matrix-markov} is therefore an oracle bound unless they are known or replaced by a valid upper bound; direct plug-in substitution yields a descriptive radius without the probability guarantee in~\eqref{eq:matrix-markov}.

\section{Two-block identifiability and stability}
\label{sec:two-block}
The effect of sampling error in the association matrix on the fitted directions depends on how clearly the population component is separated from competing components. Comparing matrix error with singular-value separation therefore connects the variability calculations to the reproducibility of fitted programs.

For unit vectors $u,v$, $\angle(u,v)=\arccos|u^\top v|\in[0,\pi/2]$ denotes the principal angle between their one-dimensional spans. For subspaces, $\|\sin\Theta(\widehat U,U)\|_\op$ denotes the sine of their largest principal angle.

Write
\begin{equation}
\Ahat=\Abar+E,
\label{eq:perturb}
\end{equation}
where $E$ is arbitrary, and let $\Abar=U\operatorname{diag}(\sigma_1,\sigma_2,\ldots)V^\top$ with $\sigma_1\ge\sigma_2\ge\cdots\ge0$. Let $m=\min(p,q)$ and restrict $1\le k\le m$. Use the zero-padding convention $\sigma_{m+1}=\widehat\sigma_{m+1}=0$ at the final singular index. Let $U_k,V_k$ contain the leading $k$ singular vectors and define
\begin{equation}
\delta_k=\sigma_k-\sigma_{k+1}.
\label{eq:gap}
\end{equation}
\begin{theorem}[Noise-to-gap stability for leading singular subspaces]
\label{thm:wedin}
If $\delta_k>0$ and $\|E\|_\op\le\delta_k/2$, then
\begin{equation}
\max\{\|\sin\Theta(\widehat U_k,U_k)\|_\op,
\|\sin\Theta(\widehat V_k,V_k)\|_\op\}
\le\frac{2\|E\|_\op}{\delta_k}.
\label{eq:wedin}
\end{equation}
\end{theorem}
\begin{proof}
If $\|E\|_\op<\delta_k/2$, Weyl's inequality~\citep{HornJohnson2013} and Wedin's operator-norm $\sin\Theta$ theorem~\citep{Wedin1972}, in the form stated by \citet[Section~2.3]{CaiZhang2018}, give the left side of~\eqref{eq:wedin} at most $\|E\|_\op/(\delta_k-\|E\|_\op)\le 2\|E\|_\op/\delta_k$. At the boundary $\|E\|_\op=\delta_k/2$, the right side of~\eqref{eq:wedin} equals one, so the conclusion holds for any choice of leading sample subspaces, including the case of a sample tie.
\end{proof}

For an individual interior singular vector, the relevant separation is two-sided, $g_k=\min\{\sigma_{k-1}-\sigma_k,\sigma_k-\sigma_{k+1}\}$, with zero singular-value spaces included when needed. Thus a leading subspace may be stable even when labels for individual nearly tied components are not.

\subsection{An observable sample-gap diagnostic}
The population separation in the stability theorem is unknown in an application. Combining the observed singular-value gap with a justified matrix-error bound gives a way to assess whether the fitted subspace is well resolved.

\begin{corollary}[Sample-gap stability bound]
Suppose $\|E\|_\op\le\varepsilon$ and let $\widehat\delta_k=\widehat\sigma_k-\widehat\sigma_{k+1}$. Then
\[
\delta_k\ge\widehat\delta_k-2\varepsilon.
\]
If $\widehat\delta_k>4\varepsilon$, then
\begin{equation}
\max\{\|\sin\Theta(\widehat U_k,U_k)\|_\op,
\|\sin\Theta(\widehat V_k,V_k)\|_\op\}
\le\frac{2\varepsilon}{\widehat\delta_k-2\varepsilon}.
\label{eq:samplegap}
\end{equation}
\end{corollary}
\begin{proof}
Weyl gives $|\widehat\sigma_j-\sigma_j|\le\varepsilon$, hence the gap lower bound. If $\widehat\delta_k>4\varepsilon$, then the resulting population gap exceeds $2\varepsilon$, so Theorem~\ref{thm:wedin} applies.
\end{proof}

For a diagnostic with a $1-\alpha$ probability guarantee, take $\varepsilon=\sqrt{V_A/\alpha}$ from Corollary~\ref{cor:matrix-rms}, including nonnull entry variances under its Gaussian assumptions.

\subsection{Kernel and PCA sensitivity}
To interpret the sensitivity of a fit to the kernel bandwidth and the retained PCA rank, we quantify how each choice changes the pairwise association matrix and compare that change with the spectral separation governing direction stability. For fixed $X,Y$, let $\Ahat(K)=X^\top KY/\kappa(K)$ with $\kappa(K)=\1^\top K\1>0$. Adding and subtracting $X^\top KY/\kappa(K')$ gives, for any $K'$ with positive mass,
\begin{equation}
\Ahat(K')-\Ahat(K)
=\frac{X^\top(K'-K)Y}{\kappa(K')}
+\frac{\kappa(K)-\kappa(K')}{\kappa(K')}\Ahat(K),
\label{eq:kernelpert}
\end{equation}
so that
\[
\|\Ahat(K')-\Ahat(K)\|_\op
\le\frac{\|X\|_\op\|K'-K\|_\op\|Y\|_\op}{\kappa(K')}
+\frac{|\kappa(K')-\kappa(K)|}{\kappa(K')}\|\Ahat(K)\|_\op.
\]
PCA truncation replaces a pair matrix $A$ by $\Pi_bA\Pi_c$, with feature-subspace projections $\Pi_b=P_bP_b^\top$ and $\Pi_c=P_cP_c^\top$. Since $A-\Pi_bA\Pi_c=(I-\Pi_b)A+\Pi_bA(I-\Pi_c)$ for any $A\in\R^{p_b\times p_c}$,
\begin{equation}
\|A-\Pi_bA\Pi_c\|_\op
\le\|(I-\Pi_b)A\|_\op+\|A(I-\Pi_c)\|_\op.
\label{eq:pcabound}
\end{equation}
Equation~\eqref{eq:pcabound} bounds the cross-type association removed by truncation, a quantity distinct from marginal explained variance.

\subsection{Estimated whitening as an additional matrix perturbation}
\label{sec:whitening-perturbation}
Fixed metrics in Proposition~\ref{prop:fixedmetric} require no new stability theory: one applies Theorem~\ref{thm:wedin} to the transformed matrix. If the marginal metrics are estimated, their errors also perturb the transformed matrix.

\begin{proposition}[Perturbation of a whitened association matrix]
\label{prop:whitening-perturbation}
Let $G_x^0,\widehat G_x\in\R^{p\times p}$ and $G_y^0,\widehat G_y\in\R^{q\times q}$ be positive definite, and let $A^0,\widehat A\in\R^{p\times q}$. Define
\begin{equation}
R^0=(G_x^0)^{-1/2}A^0(G_y^0)^{-1/2},\qquad
\widehat R=\widehat G_x^{-1/2}\widehat A\widehat G_y^{-1/2}.
\label{eq:RhatR0}
\end{equation}
Then
\begin{align}
\|\widehat R-R^0\|_\op
&\le
\|\widehat G_x^{-1/2}\|_\op\,\|\widehat A-A^0\|_\op\,\|\widehat G_y^{-1/2}\|_\op \notag\\
&\quad+
\|\widehat G_x^{-1/2}-(G_x^0)^{-1/2}\|_\op\,\|A^0\|_\op\,\|\widehat G_y^{-1/2}\|_\op \notag\\
&\quad+
\|(G_x^0)^{-1/2}\|_\op\,\|A^0\|_\op\,
\|\widehat G_y^{-1/2}-(G_y^0)^{-1/2}\|_\op.
\label{eq:whitening-general}
\end{align}
If both $G_x^0$ and $\widehat G_x$ have minimum eigenvalue at least $m_x>0$, and both $G_y^0$ and $\widehat G_y$ have minimum eigenvalue at least $m_y>0$, then, writing
\[
\varepsilon_A=\|\widehat A-A^0\|_\op,\quad
\varepsilon_x=\|\widehat G_x-G_x^0\|_\op,\quad
\varepsilon_y=\|\widehat G_y-G_y^0\|_\op,
\]
one has
\begin{equation}
\|\widehat R-R^0\|_\op
\le
\frac{\varepsilon_A}{\sqrt{m_xm_y}}
+\frac{\|A^0\|_\op\varepsilon_x}{2m_x^{3/2}m_y^{1/2}}
+\frac{\|A^0\|_\op\varepsilon_y}{2m_x^{1/2}m_y^{3/2}}.
\label{eq:whitening-simple}
\end{equation}
\end{proposition}
\begin{proof}
Add and subtract $(G_x^0)^{-1/2}A^0\widehat G_y^{-1/2}$ and $\widehat G_x^{-1/2}A^0\widehat G_y^{-1/2}$ to obtain the exact three-term decomposition underlying \eqref{eq:whitening-general}, then apply submultiplicativity.

For positive-definite $G,H$ with minimum eigenvalues at least $m$, the integral representation~\citep{Bhatia1997}
\[
G^{-1/2}=\frac{1}{\pi}\int_0^\infty t^{-1/2}(G+tI)^{-1}\,dt
\]
and the resolvent identity give
\[
\|G^{-1/2}-H^{-1/2}\|_\op
\le\frac{\|G-H\|_\op}{\pi}
\int_0^\infty\frac{t^{-1/2}}{(m+t)^2}\,dt
=\frac{\|G-H\|_\op}{2m^{3/2}}.
\]
Also $\|G^{-1/2}\|_\op\le m^{-1/2}$. Substitution into \eqref{eq:whitening-general} gives \eqref{eq:whitening-simple}.
\end{proof}

Theorem~\ref{thm:wedin} applies to $\widehat R=R^0+(\widehat R-R^0)$, so errors in both the association matrix and the marginal metrics affect stability through the whitened matrix. Original-feature coefficients also require the inverse-whitening map. For either block, write $u^0=(G^0)^{-1/2}z^0$ and $\widehat u=\widehat G^{-1/2}\widehat z$, where $z^0,\widehat z$ are unit singular directions with aligned signs. Then
\begin{align}
\|\widehat u-u^0\|_2
&\le\|\widehat G^{-1/2}\|_\op\,\|\widehat z-z^0\|_2 \notag\\
&\quad+\|\widehat G^{-1/2}-(G^0)^{-1/2}\|_\op.
\label{eq:coefficient-backmap}
\end{align}
This follows by adding and subtracting $\widehat G^{-1/2}z^0$. For aligned unit vectors, $\|\widehat z-z^0\|_2\le\sqrt2\sin\angle(\widehat z,z^0)$, and the inverse-square-root bound above controls the second term, so poor marginal conditioning can amplify coefficient error even when the whitened direction is stable. These coefficients have unit norm in their respective marginal metrics; if they are subsequently normalized in Euclidean norm, that additional map must also be controlled.

\section{Multiblock computation}
\label{sec:optimization}
For two blocks, the problem is solved exactly by a singular value decomposition. Its multiblock version is a nonconvex problem even for classical CCA, with no closed-form solution, so the fitted directions depend on an iterative procedure. Here we study the behavior of the cyclic updates used by skrCCA and derive computable upper bounds on the remaining optimization error, which the recovery guarantees of Section~\ref{sec:multiblock-identifiability} take as input.

For \eqref{eq:multiobjective}, fixing all directions except $u_b$ gives
\[
\Fhat(u_1,\ldots,u_B)=u_b^\top g_b(u)+C_b,
\qquad
g_b(u)=\sum_{c\ne b}w_{bc}\Ahat_{bc}u_c.
\]
On the unit sphere, stationarity is equivalent to $g_b(u)$ being parallel to $u_b$ for every block. These equations form a multivariate eigenvalue problem of the type studied by \citet{ChuWatterson1993}. The maximization is optimization on a product of spheres; see \citet{Absil2008}.

\subsection{Cyclic block-coordinate ascent}
Starting from a feasible point, a cyclic sweep updates
\begin{equation}
u_b\leftarrow
\begin{cases}
g_b(u)/\|g_b(u)\|_2,&g_b(u)\ne0,\\
u_b,&g_b(u)=0.
\end{cases}
\label{eq:cyclic}
\end{equation}
The update~\eqref{eq:cyclic} is the Gauss--Seidel form of Horst's algorithm~\citep{Horst1961,ChuWatterson1993}. \citet[Theorem~4.2]{ChuWatterson1993} prove monotone convergence of the simultaneous update for positive definite block matrices. Related Gauss--Seidel and alternating-variable methods have monotonicity results~\citep{ZhangChu2012,ZhangLiao2012}, as do algorithms for MAXBET, MAXDIFF, and regularized multiblock criteria~\citep{tenBergeKnol1984,tenBerge1988,HanafiKiers2006,Tenenhaus2011}. For the zero-diagonal bilinear objective here, each block subproblem is linear. The direct proof below gives a per-block version of the sweep-improvement identity of \citet[eq.~3.5 and Theorem~3.1]{ZhangLiao2012}.
\begin{theorem}[Monotonicity and nondegenerate stationary accumulation points]
For the update \eqref{eq:cyclic}: (i) every block update weakly increases $\Fhat$ and exactly maximizes it over that block with all others fixed; (ii) objective values converge; (iii) any accumulation point of sweep-end iterates at which every block gradient is nonzero is coordinatewise maximizing and stationary.
\end{theorem}
\begin{proof}
With all other blocks fixed, the problem is $\max_{\|u_b\|=1}u_b^\top g_b+C_b$, whose maximum is $\|g_b\|$ and is attained by \eqref{eq:cyclic} when $g_b\ne0$. Thus the objective is monotone. Compactness of the feasible product of spheres implies boundedness, hence convergence of objective values. Consider sweep-end iterates $u^{t_\ell}\to u^\star$, equivalently the starts of the following sweeps, with all $g_b(u^\star)\ne0$. Every ensuing per-block objective increment tends to zero. The exact block-improvement identity
\[
\Delta_{t,b}=\frac{\|g_b\|}{2}\|u_b^{\rm new}-u_b^{\rm old}\|_2^2
\]
first forces the movement of block 1 to vanish, because its gradient converges to $g_1(u^\star)\ne0$. The iterate after that update still converges to $u^\star$, so block 2's gradient also converges to its nonzero limiting value; induction gives the same conclusion through the entire sweep. Continuity of each normalized update now yields $u_b^\star=g_b(u^\star)/\|g_b(u^\star)\|$, proving coordinatewise optimality and stationarity.
\end{proof}

\subsection{Computable upper bounds on optimization error}
Upper bounds on the global optimum give a computable measure of the remaining optimization error.

Let $\mathcal A$ be the symmetric block matrix with zero diagonal blocks and off-diagonal blocks $\mathcal A_{bc}=w_{bc}\Ahat_{bc}$. For stacked $u$, $\|u\|_2^2=B$ and $2\Fhat(u)=u^\top\mathcal Au$.
\begin{proposition}[Objective upper bounds]
\label{prop:upperbounds}
Let $\Fhat^\star$ be the global optimum. Then
\begin{equation}
\Fhat^\star\le U_{\rm spec}:=\frac{B}{2}\lambda_{\max}(\mathcal A),
\qquad
\Fhat^\star\le U_{\rm pair}:=\sum_{b<c}w_{bc}\sigma_1(\Ahat_{bc}).
\label{eq:upperbounds}
\end{equation}
For any feasible $\widetilde u$,
\begin{equation}
\overline\eta(\widetilde u)=\min(U_{\rm spec},U_{\rm pair})-\Fhat(\widetilde u)
\label{eq:etaub}
\end{equation}
is a nonnegative upper bound on $\Fhat^\star-\Fhat(\widetilde u)$.
\end{proposition}
\begin{proof}
The spectral bound follows from the Rayleigh quotient and $\|u\|_2^2=B$. The pairwise bound follows from $u_b^\top\Ahat_{bc}u_c\le\sigma_1(\Ahat_{bc})$ for each pair.
\end{proof}
The block-multiplier certificate of \citet[Result~1]{HanafitenBerge2003} gives another sufficient condition for global optimality. At a stationary point $u$, let $\lambda_b=u_b^\top(\mathcal Au)_b$. If $\operatorname{diag}(\lambda_b I_{p_b})-\mathcal A$ is positive semidefinite, then every feasible $v$ satisfies $\Fhat(v)\le\tfrac12\sum_b\lambda_b=\Fhat(u)$.

\section{Multiblock population identifiability and stability}
\label{sec:multiblock-identifiability}
A common program across cell types requires pairwise associations that support compatible feature directions. Population models make this compatibility precise, and the resulting bounds relate recovery of a shared program to association strength, connections among types, matrix error, and numerical accuracy.

The results below take pairwise matrix errors as given and do not require the Gaussian null of Section~\ref{sec:null}.

\subsection{Coherent rank-one population model}
The simplest shared program assigns one direction to each type and uses that direction in every association involving the type. This model isolates how the strength and pattern of pairwise connections identify the directions and their relative orientation.

\begin{assumption}[Coherent rank-one pair matrices]
\label{ass:coherent}
For every pair with $w_{bc}>0$,
\begin{equation}
\Abar_{bc}=\theta_{bc}\lambda_b\lambda_c^\top,
\qquad \theta_{bc}>0,\qquad \|\lambda_b\|_2=1.
\label{eq:coherent}
\end{equation}
\end{assumption}
The assumption requires edge signs that can be made simultaneously positive by orienting the blocks; arbitrary signed rank-one pair matrices need not have this property. Let the graph contain edge $\{b,c\}$ when $w_{bc}\theta_{bc}>0$, and define
\begin{equation}
a_{bc}=w_{bc}\theta_{bc},\qquad D_b=\sum_{c\ne b}a_{bc}.
\label{eq:degrees}
\end{equation}
For $t_b=\lambda_b^\top u_b$, the population objective is $F^0(u)=\sum_{b<c}a_{bc}t_bt_c$.

\begin{theorem}[Population identifiability and exact loss decomposition]
Under Assumption~\ref{ass:coherent},
\begin{equation}
F^{0\star}-F^0(u)
=\frac12\sum_{b<c}a_{bc}
\{(t_b-t_c)^2+(1-t_b^2)+(1-t_c^2)\},
\label{eq:loss}
\end{equation}
and therefore
\begin{equation}
F^{0\star}-F^0(u)
\ge\frac12\sum_{b=1}^BD_b\sin^2\angle(u_b,\lambda_b).
\label{eq:quadgrowth}
\end{equation}
If the graph is connected, the population maximizers are exactly the common positive and common negative sign choices. On a disconnected graph, each nontrivial connected component may choose its sign independently; an isolated vertex is not identified.
\end{theorem}
\begin{proof}
Use the identity $2(1-t_bt_c)=(t_b-t_c)^2+(1-t_b^2)+(1-t_c^2)$ edge by edge. Equality at the optimum requires $t_bt_c=1$ on every positive edge, which forces $t_b=t_c=\pm1$ and propagates a common sign through each connected component.
\end{proof}
The graph has one vertex per block, not per spatial observation. In CoPro's all-pairs objective it is complete under Assumption~\ref{ass:coherent}, so the only ambiguity is a simultaneous reversal of every block; with selected pairs, the connected components reverse independently. In the scalar case $p_b=1$, after orienting each block so that $\lambda_b=1$, the population block matrix is entrywise nonnegative, and uniqueness of the maximizer up to a common sign for nonnegative irreducible block matrices is a Perron--Frobenius-type result~\citep{ZhangLiaoSun2011}; see also \citet{HanafitenBerge2003}. This is the noiseless case of $\mathbb Z_2$ synchronization, where node signs are recoverable up to a global flip exactly when the graph is connected~\citep{Abbe2014}.

\subsection{Stability of an approximate sample optimizer}
A computed fit differs from the population optimum through both sampling variation and incomplete optimization. Write
\begin{equation}
\Ahat_{bc}=\Abar_{bc}+E_{bc},\qquad
\Delta=\sum_{b<c}w_{bc}\|E_{bc}\|_\op.
\label{eq:Delta}
\end{equation}
\begin{theorem}[Matrix perturbation and optimization error decomposition]
\label{thm:multistability}
Let $\widetilde u$ satisfy $\Fhat^\star-\Fhat(\widetilde u)\le\eta$ for some $\eta\ge0$. Under Assumption~\ref{ass:coherent},
\begin{equation}
\sum_{b=1}^BD_b\sin^2\angle(\widetilde u_b,\lambda_b)
\le4\Delta+2\eta.
\label{eq:multistability}
\end{equation}
The unknown $\eta$ may be replaced by \eqref{eq:etaub}.
This bound controls the unoriented direction in each block: its loss is unchanged by a separate sign flip of any block. Relative orientation is treated in the remark below.
\end{theorem}
\begin{proof}
Uniformly over feasible $u$, $|\Fhat(u)-F^0(u)|\le\Delta$. If $u^0$ is a population maximizer,
\[
F^{0\star}-F^0(\widetilde u)
\le \Fhat^\star-\Fhat(\widetilde u)+2\Delta
\le\eta+2\Delta.
\]
Combine with \eqref{eq:quadgrowth} and multiply by two.
\end{proof}

Under a stochastic model, the exact second moments of Section~\ref{sec:matrix-noise} supply a single event on which every pairwise error is controlled, turning Theorem~\ref{thm:multistability} into a probability statement.

\begin{corollary}[High-probability recovery from pairwise second moments]
\label{cor:hp-recovery}
Conditional on the design, let the pair matrices $\Ahat_{bc}$ be random with $\E\Ahat_{bc}=\Abar_{bc}$ and finite $V_{bc}=\E\|\Ahat_{bc}-\Abar_{bc}\|_\F^2$ for every pair with $w_{bc}>0$. Define
\begin{equation}
S=\sum_{b<c}w_{bc}\sqrt{V_{bc}},\qquad r_\alpha=\frac{S}{\sqrt\alpha},\qquad 0<\alpha<1.
\label{eq:r-alpha}
\end{equation}
Then $\Pr(\Delta\le r_\alpha)\ge1-\alpha$ for $\Delta$ in~\eqref{eq:Delta}, and $\Delta=0$ almost surely if $S=0$. No independence between the errors of different pairs is required. Consequently, under Assumption~\ref{ass:coherent}, with probability at least $1-\alpha$ every feasible $\widetilde u$ satisfies
\begin{equation}
\sum_{b=1}^BD_b\sin^2\angle(\widetilde u_b,\lambda_b)\le4r_\alpha+2\overline\eta(\widetilde u),
\label{eq:hp-recovery}
\end{equation}
with the computable upper gap $\overline\eta(\widetilde u)$ of~\eqref{eq:etaub}. Under joint mean-zero Gaussianity of the blocks, $V_{bc}$ is given exactly by Proposition~\ref{prop:nonnull-variance} applied entrywise.
\end{corollary}
\begin{proof}
By Minkowski's inequality in $L^2$ and $\|E\|_\op\le\|E\|_\F$,
\[
\{\E\Delta^2\}^{1/2}\le\sum_{b<c}w_{bc}\{\E\|E_{bc}\|_\op^2\}^{1/2}\le\sum_{b<c}w_{bc}\sqrt{V_{bc}}=S .
\]
Markov's inequality applied to $\Delta^2$ gives $\Pr(\Delta>r_\alpha)\le S^2/r_\alpha^2=\alpha$; if $S=0$ then $\E\Delta^2=0$. On the event $\{\Delta\le r_\alpha\}$, Theorem~\ref{thm:multistability} with $\eta=\overline\eta(\widetilde u)\ge\Fhat^\star-\Fhat(\widetilde u)$ gives~\eqref{eq:hp-recovery}.
\end{proof}

The spatial design, kernel, and covariance model determine the second moments $V_{bc}$; these determine $r_\alpha$; the fit determines $\overline\eta$; and the population model supplies the degrees $D_b$ against which the right side is compared. The bound is informative only when $4r_\alpha+2\overline\eta<\sum_bD_b$. The single event $\{\Delta\le r_\alpha\}$ is reused for the orientation and deflation statements below. The null-only radius of Theorem~\ref{thm:operator-concentration} must not be substituted for $r_\alpha$ in a nonnull recovery claim.

\begin{remark}[Relative orientation and weak graph connections]
Writing $t_b=\lambda_b^\top\widetilde u_b$, the same proof and the full identity~\eqref{eq:loss} also give
\begin{equation}
\sum_{b<c}a_{bc}(t_b-t_c)^2+
\sum_bD_b(1-t_b^2)\le4\Delta+2\eta.
\label{eq:orientation-energy}
\end{equation}
The first term retains information about relative orientation. Connectedness fixes the relative signs at an exact population optimum, leaving one common sign undetermined, but a weak graph cut can make opposing orientations nearly optimal. For example, let $a_{12}=a_{34}=1$ and $a_{23}=\varepsilon>0$, with no other edges. The directions $(\lambda_1,\lambda_2,-\lambda_3,-\lambda_4)$ have zero per-block sine loss but population objective loss $2\varepsilon$, although all weighted degrees are at least one. The relative-sign problem is a synchronization problem over the block graph~\citep{Singer2011,BandeiraSingerSpielman2013}.
\end{remark}

\begin{corollary}[Recovery up to one common sign]
\label{cor:common-sign}
Under Assumption~\ref{ass:coherent} with a connected graph, let $L_a$ be the Laplacian of the block graph with edge weights $a_{bc}$, let $\lambda_2(L_a)>0$ be its algebraic connectivity, and put $D_{\min}=\min_bD_b$ and $\mu=\min\{D_{\min},\lambda_2(L_a)\}$. If $\Fhat^\star-\Fhat(\widetilde u)\le\eta$ and $R=4\Delta+2\eta$, then, with $t_b=\lambda_b^\top\widetilde u_b$ and $\bar t=B^{-1}\sum_bt_b$,
\begin{equation}
\min_{s\in\{-1,1\}}\sum_{b=1}^B\|\widetilde u_b-s\lambda_b\|_2^2
=2B(1-|\bar t|)\le\min\Bigl\{2B,\frac{2R}{\mu}\Bigr\}.
\label{eq:common-sign}
\end{equation}
If $R<\mu$, all $t_b$ have the same nonzero sign. Replacing $\Delta$ by $r_\alpha$ and $\eta$ by $\overline\eta(\widetilde u)$ gives the same conclusions with probability at least $1-\alpha$ under the conditions of Corollary~\ref{cor:hp-recovery}.
\end{corollary}
\begin{proof}
With $t=(t_1,\ldots,t_B)^\top$, the first term of~\eqref{eq:orientation-energy} is $t^\top L_at=\sum_{b<c}a_{bc}(t_b-t_c)^2$. On a connected graph the kernel of $L_a$ is spanned by $\1$, so $t^\top L_at\ge\lambda_2(L_a)\sum_b(t_b-\bar t)^2$. Using $D_b\ge D_{\min}$ and $1-t_b^2\ge0$ in~\eqref{eq:orientation-energy},
\[
\mu\Bigl\{\sum_b(t_b-\bar t)^2+\sum_b(1-t_b^2)\Bigr\}=\mu B(1-\bar t^2)\le R,
\]
because $\sum_b(t_b-\bar t)^2=\sum_bt_b^2-B\bar t^2$. For unit vectors, $\|\widetilde u_b-s\lambda_b\|_2^2=2-2st_b$, so the minimum over $s$ of the sum is $2B-2|\sum_bt_b|=2B(1-|\bar t|)$, and $1-|\bar t|\le1-\bar t^2$ because $|\bar t|\le1$. The bound $2B$ always holds. If $R<\mu$, the aligned total is below two, whereas a block with $st_b\le0$ alone contributes at least two.
\end{proof}

The bound separates two mechanisms: blockwise direction loss is governed by the degrees $D_b$ in Theorem~\ref{thm:multistability}, whereas orientation up to one common sign is additionally governed by $\lambda_2(L_a)$. In the example above, with $a_{12}=a_{34}=1$ and $a_{23}=\varepsilon$, every degree is at least one but the test vector $(1,1,-1,-1)/2$ gives $\lambda_2(L_a)\le\varepsilon$, so $\mu\le\varepsilon$. Substantial degrees therefore do not protect relative orientation across a weak cut.

\subsection{A first-order bound at stationary points}
\label{sec:first-order}
The objective-comparison bound controls the squared sine by the matrix error, so the direction error itself scales like the square root of that error. At stationary points with positive block multipliers, the first-order condition supplies a linear bound with an observable normalizer.

Recall $g_b(u)=\sum_{c\ne b}w_{bc}\Ahat_{bc}u_c$ from Section~\ref{sec:optimization}, and write $P_b^\perp=I-\lambda_b\lambda_b^\top$ and $\Delta_b=\sum_{c\ne b}w_{bc}\|E_{bc}\|_\op$, so that $\sum_b\Delta_b=2\Delta$.
\begin{proposition}[First-order perturbation bound]
\label{prop:first-order}
Under Assumption~\ref{ass:coherent}, let $\widetilde u$ satisfy the stationarity condition $g_b(\widetilde u)=\nu_b\widetilde u_b$ for block $b$ with $\nu_b=\widetilde u_b^\top g_b(\widetilde u)>0$. Then
\begin{equation}
\sin\angle(\widetilde u_b,\lambda_b)
\le\frac{1}{\nu_b}\sum_{c\ne b}w_{bc}\|P_b^\perp E_{bc}\|_\op
\le\frac{\Delta_b}{\nu_b}.
\label{eq:first-order}
\end{equation}
With $t_c=\lambda_c^\top\widetilde u_c$, the multiplier satisfies $\nu_b\ge t_b\sum_{c\ne b}a_{bc}t_c-\Delta_b$. In particular, if all $t_c$ have a common sign and $\min_c|t_c|\ge t_0>0$, then $\nu_b\ge t_0^2D_b-\Delta_b$.
\end{proposition}
\begin{proof}
Substituting $\Ahat_{bc}=\theta_{bc}\lambda_b\lambda_c^\top+E_{bc}$ gives
$g_b(\widetilde u)=\lambda_b\sum_{c\ne b}a_{bc}t_c+\sum_{c\ne b}w_{bc}E_{bc}\widetilde u_c$.
Applying $P_b^\perp$ to the stationarity condition removes the first term, so
$\nu_bP_b^\perp\widetilde u_b=\sum_{c\ne b}w_{bc}P_b^\perp E_{bc}\widetilde u_c$. For a unit vector, $\|P_b^\perp\widetilde u_b\|_2=\sin\angle(\widetilde u_b,\lambda_b)$, and the triangle inequality gives the first bound in~\eqref{eq:first-order}; the second uses $\|P_b^\perp\|_\op\le1$. Taking the inner product of the stationarity condition with $\widetilde u_b$ gives $\nu_b=t_b\sum_{c\ne b}a_{bc}t_c+\widetilde u_b^\top\sum_{c\ne b}w_{bc}E_{bc}\widetilde u_c$, and the last term is at least $-\Delta_b$.
\end{proof}

Every fixed point of the cyclic update~\eqref{eq:cyclic} with nonzero block gradients satisfies the hypothesis with $\nu_b=\|g_b(\widetilde u)\|_2$, which the algorithm computes. The proposition requires neither global optimality nor an optimization gap: it holds at any stationary point whose multipliers are positive. For such a stationary point, if $R=4\Delta+2\eta<\mu$, Corollary~\ref{cor:common-sign} gives a common sign and Theorem~\ref{thm:multistability} gives $t_b^2\ge1-R/D_{\min}$ for every block, so that
\begin{equation}
\sin\angle(\widetilde u_b,\lambda_b)\le\frac{\sum_{c\ne b}w_{bc}\|P_b^\perp E_{bc}\|_\op}{(1-R/D_{\min})D_b-\Delta_b}
\label{eq:first-order-apriori}
\end{equation}
whenever the denominator is positive. This is a noise-to-degree bound of the same linear form as the two-block noise-to-gap bound of Theorem~\ref{thm:wedin}, with the weighted degree of the block in the role of the singular gap. The numerator involves only the rows of $E_{bc}$ orthogonal to $\lambda_b$: a perturbation aligned with the population direction changes the multiplier but not the direction.

\begin{remark}[Aligned and orthogonal error]
\label{rem:aligned-error}
Under Assumption~\ref{ass:coherent}, write $E_{bc}=\epsilon_{bc}\lambda_b\lambda_c^\top+E_{bc}^\perp$ with $\epsilon_{bc}=\lambda_b^\top E_{bc}\lambda_c$, so that $\Ahat_{bc}=(\theta_{bc}+\epsilon_{bc})\lambda_b\lambda_c^\top+E_{bc}^\perp$. If $\theta_{bc}+\epsilon_{bc}>0$ on every edge, the first term satisfies Assumption~\ref{ass:coherent} with the same directions and with edge weights $\widetilde a_{bc}=w_{bc}(\theta_{bc}+\epsilon_{bc})$, so Theorem~\ref{thm:multistability} and Corollary~\ref{cor:common-sign} hold with $\Delta^\perp=\sum_{b<c}w_{bc}\|E^\perp_{bc}\|_\op$ in place of $\Delta$ and the realized degrees and algebraic connectivity of $\widetilde a$ in place of those of $a$. Because $\|E_{bc}\|_\F^2=\epsilon_{bc}^2+\|E_{bc}^\perp\|_\F^2$,
\[
\E\|E^\perp_{bc}\|_\F^2=V_{bc}-\Var(\lambda_b^\top\Ahat_{bc}\lambda_c),
\qquad
\E\|P_b^\perp E_{bc}\|_\F^2=V_{bc}-\sum_{j=1}^{p_c}\Var(\lambda_b^\top\Ahat_{bc}e_j),
\]
where the subtracted terms are fixed-direction variances given by Proposition~\ref{prop:nonnull-variance}. The same second-moment argument bounds $\Delta^\perp$ and the numerator in Proposition~\ref{prop:first-order}. Recovery using the realized edge weights remains subject to their positivity.
\end{remark}

\section{Higher-order components and deflation}
\label{sec:deflation}
Sequential extraction seeks additional spatial programs after removing the directions already fitted. Its interpretation depends on the structure of the remaining population association, and its accuracy depends on how errors in earlier directions alter the residual matrices.

\subsection{Projected and constrained deflation}
Suppose components $1,\ldots,k-1$ have been obtained with orthonormal directions within each block. Let $U_b^{(k-1)}=[u_b^{(1)},\ldots,u_b^{(k-1)}]$ and $\Pi_b^{(k)}=I-U_b^{(k-1)}U_b^{(k-1)\top}$. Projected deflation replaces each pair matrix by
\begin{equation}
\Ahat_{bc}^{(k)}=\Pi_b^{(k)}\Ahat_{bc}\Pi_c^{(k)}.
\label{eq:deflated}
\end{equation}
This is projection deflation~\citep{Mackey2008}. \citet[p.~97]{HanafitenBerge2003} use the same coefficient-space projection for successive MAXBET components. For broader treatments of successive multiblock extraction, see \citet{tenBerge1988} and \citet{Tenenhaus2017}.
\begin{proposition}[Value equivalence of projected and constrained deflation]
\label{prop:deflation-value}
Assume each deflated subspace is nontrivial. The projected program and the program that explicitly constrains every $u_b$ to the deflated subspace have the same optimal value. Every projected-program maximizer with nonzero projected block gradients lies in the deflated subspaces. The two programs need not have identical stationary-point sets when a projected block gradient is zero.
\end{proposition}
\begin{proof}
Project each feasible vector into its deflated subspace. The objective depends only on those projections. Maximizing the resulting multilinear function over products of closed unit balls admits a maximizer on the boundary of every nontrivial block: if a block gradient is nonzero, align the block with it; if the gradient is zero, any unit vector in that subspace preserves the objective. This gives equality of optimal values. At a projected maximizer with nonzero block gradient, the exact block maximizer lies in the projected range.
\end{proof}

\subsection{Population recovery under ordered factors}
Recovering named factors in a sequence requires a population ordering that is shared across type pairs. A model with common factor strengths and pair-specific overall weights makes this ordering explicit and reveals the separation needed to recover each successive component.

\begin{assumption}[Ordered separable latent factors]
\label{ass:ordered-factors}
For $j=1,\ldots,r$, let $\lambda_{b,j}$ be orthonormal within each block. Suppose
\begin{equation}
\Abar_{bc}=\omega_{bc}\sum_{j=1}^r\theta_j\lambda_{b,j}\lambda_{c,j}^\top,
\qquad \theta_1>\theta_2>\cdots>\theta_r>0,
\label{eq:separablefactors}
\end{equation}
where $\omega_{bc}\ge0$ and the graph with edges $w_{bc}\omega_{bc}>0$ is connected. Define $d_b=\sum_{c\ne b}w_{bc}\omega_{bc}$ and $\theta_{r+1}=0$.
\end{assumption}

\begin{theorem}[Population recovery by sequential deflation]
\label{thm:ordered-deflation}
Under Assumption~\ref{ass:ordered-factors}, suppose the true factors $\lambda_{b,1},\ldots,\lambda_{b,k-1}$ have been deflated. Then the order-$k$ population problem is maximized exactly by the common positive or common negative $k$th factor. Moreover,
\begin{equation}
F_k^{0\star}-F_k^0(u)
\ge\frac12(\theta_k-\theta_{k+1})
\sum_{b=1}^Bd_b\sin^2\angle(u_b,\lambda_{b,k}).
\label{eq:deflation-gap}
\end{equation}
\end{theorem}
\begin{proof}
After earlier factors are removed, write $t_{b,j}=\lambda_{b,j}^\top u_b$. Orthonormality gives $\sum_{j\ge k}t_{b,j}^2\le1$. Applying $2xy\le x^2+y^2$ edgewise and using $\theta_j\le\theta_{k+1}$ for $j>k$ yields \eqref{eq:deflation-gap}. Equality is attained at the $k$th factor, and graph connectedness synchronizes the signs.
\end{proof}

\subsection{Error propagation with estimated directions}
\begin{remark}[Deflation with estimated directions]
In a fitted analysis, the directions removed at each stage are themselves estimates, so their error propagates to the residual matrices of later stages. Let $\widehat\Pi_b^{(k)}$ and $\Pi_b^{(k)}$ be estimated and population deflation projectors. For any pair matrix $A$,
\begin{equation}
\|\widehat\Pi_b^{(k)}A\widehat\Pi_c^{(k)}-\Pi_b^{(k)}A\Pi_c^{(k)}\|_\op
\le
\{\|\widehat\Pi_b^{(k)}-\Pi_b^{(k)}\|_\op+
\|\widehat\Pi_c^{(k)}-\Pi_c^{(k)}\|_\op\}\|A\|_\op.
\label{eq:deflprop}
\end{equation}
For projectors of equal rank, the projector difference equals the sine of the largest principal angle between the corresponding subspaces. Earlier component error therefore creates an additional pairwise matrix perturbation at order $k$. Because the residual population matrices generally contain several factors, the rank-one bound of Theorem~\ref{thm:multistability} is replaced below by the ordered-factor gap of Theorem~\ref{thm:ordered-deflation}.
\end{remark}
\begin{proof}[Verification of \eqref{eq:deflprop}]
Add and subtract $\widehat\Pi_b^{(k)}A\Pi_c^{(k)}$ and use the triangle inequality, submultiplicativity, and $\|\Pi\|_\op\le1$ for orthogonal projectors.
\end{proof}

\begin{corollary}[Recovery with estimated deflation]
\label{cor:estimated-deflation}
Under Assumption~\ref{ass:ordered-factors}, let $1\le k\le r$. Let $\Pi_b^{(k)}$ remove the true first $k-1$ factors, and let $\widehat\Pi_b^{(k)}$ be any orthogonal projector of the same rank. Define
\begin{align}
A^0_{bc,k}&=\Pi_b^{(k)}\Abar_{bc}\Pi_c^{(k)},&
\widehat A_{bc,k}&=\widehat\Pi_b^{(k)}\Ahat_{bc}\widehat\Pi_c^{(k)},\notag\\
\Delta_k&=\sum_{b<c}w_{bc}\|\widehat A_{bc,k}-A^0_{bc,k}\|_\op,&
e_{b,k}&=\|\widehat\Pi_b^{(k)}-\Pi_b^{(k)}\|_\op.
\label{eq:deflation-error-total}
\end{align}
Let $\widehat F_k$ be the bilinear objective formed from $\widehat A_{bc,k}$ over the product of unit spheres. If
$\widehat F_k^\star-\widehat F_k(\widetilde u^{(k)})\le\eta_k$, then
\begin{equation}
\sum_b d_b\sin^2\angle(\widetilde u_b^{(k)},\lambda_{b,k})
\le\frac{4\Delta_k+2\eta_k}{\theta_k-\theta_{k+1}},
\label{eq:estimated-deflation-recovery}
\end{equation}
where
\begin{equation}
\Delta_k\le\sum_{b<c}w_{bc}
\left\{\|E_{bc}\|_\op+(e_{b,k}+e_{c,k})\|\Abar_{bc}\|_\op\right\}.
\label{eq:estimated-deflation-input}
\end{equation}
The conclusion also applies to an approximate optimizer constrained to the estimated deflated subspaces, with its gap measured against the constrained optimum. It concerns unoriented block directions. Applying Proposition~\ref{prop:upperbounds} to the deflated sample matrices $\widehat A_{bc,k}$ supplies a computable upper bound $\overline\eta_k$ for $\eta_k$.
\end{corollary}
\begin{proof}
Expand the residual error as
\[
\widehat A_{bc,k}-A^0_{bc,k}
=\widehat\Pi_b^{(k)}E_{bc}\widehat\Pi_c^{(k)}
+\widehat\Pi_b^{(k)}\Abar_{bc}\widehat\Pi_c^{(k)}
-\Pi_b^{(k)}\Abar_{bc}\Pi_c^{(k)}.
\]
Projector contraction and~\eqref{eq:deflprop} give~\eqref{eq:estimated-deflation-input}. Uniformly over the common product of spheres, $|\widehat F_k-F_k^0|\le\Delta_k$, hence
$F_k^{0\star}-F_k^0(\widetilde u^{(k)})\le2\Delta_k+\eta_k$.
The growth inequality~\eqref{eq:deflation-gap} holds on this full product as well: components in the removed directions contribute nothing, and the remaining squared factor coordinates still sum to at most one. Combining the inequalities proves~\eqref{eq:estimated-deflation-recovery}. Proposition~\ref{prop:deflation-value} identifies the projected and constrained sample optimal values, proving the last assertion. No independence between estimated projectors and matrix errors is needed for this deterministic result.
\end{proof}
Error propagation through inexact Hotelling deflation has been analyzed for symmetric positive semidefinite matrices~\citep{Liao2024}; rectangular pair matrices and multiblock objectives lie outside that analysis.

\begin{remark}[High-probability form]
\label{rem:deflation-hp}
On the event $\{\Delta\le r_\alpha\}$ of Corollary~\ref{cor:hp-recovery}, inequality~\eqref{eq:estimated-deflation-input} gives
\begin{equation}
\Delta_k\le r_\alpha+\sum_{b<c}w_{bc}(e_{b,k}+e_{c,k})\|\Abar_{bc}\|_\op
\label{eq:deflation-hp}
\end{equation}
simultaneously for every stage $k$, because the event concerns the original pair matrices only.
\end{remark}

\begin{remark}[An explicit conservative recursion]
\label{rem:deflation-recursion}
The following recursion tracks the projector-error bound across successive stages. If the previously estimated axes are orthonormal within each block, then
\begin{equation}
e_{b,k}^2\le\sum_{j<k}\sin^2\angle(\widetilde u_b^{(j)},\lambda_{b,j}),
\qquad e_{b,k}\le1.
\label{eq:projector-rss}
\end{equation}
Indeed, let $U_b$ and $\widehat U_b$ contain the true and estimated first $k-1$ directions, each with orthonormal columns. The squared operator norm of the projector difference is the largest squared principal-angle sine between the two subspaces, which is at most the sum of all squared sines, $k-1-\|U_b^\top\widehat U_b\|_\F^2$; keeping only the diagonal entries of $U_b^\top\widehat U_b$ in the Frobenius norm gives~\eqref{eq:projector-rss}. The root-sum-square form improves the linear bound $\sum_{j<k}\sin\angle(\widetilde u_b^{(j)},\lambda_{b,j})$ obtained from the triangle inequality on rank-one projectors, and the cap by one holds for any two orthogonal projectors of equal rank. To close the recursion, suppose $\|\Abar_{bc}\|_\op\le M_{bc}$ and work on the event $\{\Delta\le r_\alpha\}$ of Remark~\ref{rem:deflation-hp}; any deterministic bound on $\Delta$ may replace $r_\alpha$. Starting with $\overline e_{b,1}=0$, set
\begin{align}
\overline\Delta_k&=r_\alpha+\sum_{b<c}w_{bc}(\overline e_{b,k}+\overline e_{c,k})M_{bc},\notag\\
q_{b,k}&=\min\left\{1,\sqrt{\frac{4\overline\Delta_k+2\overline\eta_k}
{d_b(\theta_k-\theta_{k+1})}}\right\},\notag\\
\overline e_{b,k+1}&=\min\Bigl\{1,\Bigl(\sum_{j\le k}q_{b,j}^2\Bigr)^{1/2}\Bigr\}.
\label{eq:deflation-recursion}
\end{align}
Induction using Corollary~\ref{cor:estimated-deflation} and~\eqref{eq:projector-rss} gives $\sin\angle(\widetilde u_b^{(k)},\lambda_{b,k})\le q_{b,k}$ and $e_{b,k}\le\overline e_{b,k}$ for every stage on that event. For pairs with $w_{bc}>0$, on the same event one may take $M_{bc}=\|\Ahat_{bc}\|_\op+r_\alpha/w_{bc}$. The recursion saturates at $q_{b,k}=1$ once $4\overline\Delta_k+2\overline\eta_k\ge d_b(\theta_k-\theta_{k+1})$, so the bound for that block and stage is uninformative.
\end{remark}

\subsection{Limitations and score interpretation}
Without a shared ordering, the first component can mix factors.

\begin{example}[A nonseparable population produces a mixed component]
Let $B=3$ and consider two orthogonal factors. Factor 1 couples only blocks $(1,2)$ with strength $1$, and factor 2 couples only blocks $(2,3)$ with strength $1-\varepsilon$, where $0<\varepsilon<1$. Set
\[
u_1=\lambda_{1,1},\qquad u_3=\lambda_{3,2},\qquad
u_2=\alpha\lambda_{2,1}+\beta\lambda_{2,2},\quad \alpha^2+\beta^2=1.
\]
The objective is $\alpha+(1-\varepsilon)\beta$, whose maximum is $\sqrt{1+(1-\varepsilon)^2}>1$. This is the global maximum: for fixed $u_2$, the optimal $u_1,u_3$ align in sign with the two relevant projections, reducing the problem to maximizing $|\alpha|+(1-\varepsilon)|\beta|$ over $\alpha^2+\beta^2\le1$; Cauchy--Schwarz gives exactly $\sqrt{1+(1-\varepsilon)^2}$. Hence the first population component is mixed in block 2, and subsequent deflation depends on that mixed direction.
\end{example}

\begin{remark}[Coefficient orthogonality and score covariance]
Within-block orthogonality prevents reuse of previously selected coefficient directions. For centered $X_b$, however, the empirical covariance of two scores is
\[
\frac{1}{n_b}(X_bu_b^{(j)})^\top(X_bu_b^{(k)})
=u_b^{(j)\top}S_bu_b^{(k)},\qquad S_b=X_b^\top X_b/n_b.
\]
Euclidean coefficient orthogonality makes this zero when $S_b$ is proportional to the identity, as after whitening, for any number of blocks. Without whitening, the required condition is orthogonality in the metric $S_b$. A separate property holds for two blocks: successive singular pairs satisfy $u_j^\top\Ahat v_k=0$ for $j\ne k$, so distinct axes have zero kernel-weighted cross-association between the blocks. With three or more blocks, a joint fit need not diagonalize every pair matrix, even when each block is whitened.
\end{remark}

\section{Practitioner diagnostics and failure modes}
\label{sec:diagnostics}
When fitted programs vary across analyses, the useful next step depends on the source of the variation: sampling error, weak spectral separation, numerical suboptimality, and preprocessing choices call for different checks. Table~\ref{tab:diagnostics} relates each practical question to the theoretical quantity that addresses it.

\begin{table}[!htb]
\caption{Practitioner diagnostics.}
\label{tab:diagnostics}
\small
\begin{tabularx}{\textwidth}{>{\raggedright\arraybackslash}p{0.22\textwidth} >{\raggedright\arraybackslash}p{0.37\textwidth} X}
\toprule
Question & Suggested diagnostic & Interpretation\\
\midrule
Is a fixed pairwise association precise? & For fixed directions and justified score covariances, use Theorem~\ref{thm:variance}. Check reproducibility on held-out specimens; distinguish technical sections from independent biological replicates. & Small $\neff$ means high null variance; low stable rank warns against a Gaussian approximation. Plug-in calibration requires additional assumptions.\\
Is a two-block axis stable? & Report $s_A$ and a justified error radius as in Section~\ref{sec:matrix-noise}; compare the radius with the observed gap. & Apply~\eqref{eq:samplegap} with that radius.\\
Is a multiblock fit numerically reproducible? & Inspect monotone objective traces, multiple starts, stationarity residuals, and the upper gap in \eqref{eq:etaub}. & Agreement across starts supports numerical reproducibility; the upper gap bounds suboptimality.\\
Is a multiblock axis recoverable, and is its orientation consistent across types? & Report $\overline\eta$ and fitted multipliers $\nu_b$. If justified second moments and population-model quantities are available, assess Corollaries~\ref{cor:hp-recovery}--\ref{cor:common-sign}; use~\eqref{eq:first-order} at stationary points with positive multipliers. & Low population degree weakens direction recovery; weak graph connections weaken relative orientation. The degrees and graph connectivity are population quantities. Plug-in diagnostics do not inherit the stated probability guarantees without further justification.\\
How does imbalance enter? & Separate pair scaling, pairwise variance, and spatial coverage. & Uniform duplication is a scaling issue; rare or redundant blocks can still be less precise; selective coverage changes the target.\\
How sensitive is the fit to PCA or standardization? & Refit over plausible PCA ranks; for standardized criteria inspect conditioning and, when relevant, the terms in \eqref{eq:whitening-simple}. & Basis rotations within a fixed subspace are irrelevant; rank and whitening estimation can matter.\\
How sensitive is the fit to the kernel? & Compare matrices and fitted subspaces over scientifically prespecified kernels. & Kernel changes can perturb estimation and also change the population spatial scale being summarized.\\
Are higher components separately interpretable? & Inspect adjacent gaps, specimen-level reproducibility, and accumulated projector error in~\eqref{eq:estimated-deflation-input}, tracked by the recursion~\eqref{eq:deflation-recursion}; compare estimated-projector fits with refits that hold earlier axes fixed. & Uncorrelated within-type scores require covariance-metric orthogonality. Small factor gaps weaken guarantees for individual axes.\\
\bottomrule
\end{tabularx}
\end{table}

\section{Discussion}
\label{sec:discussion}
The covariance form of skrCCA targets spatially weighted cross-type association at the observed sites. Kernel choice changes the spatial scale of that target, while feature reduction changes the retained representation. Stability requires more than a large fitted association: two-block recovery depends on a justified matrix-error bound relative to spectral separation, and multiblock recovery also depends on numerical suboptimality and population structure. The coherent rank-one and ordered-factor models separate four failure modes: low weighted degree, weak graph cuts, nearly tied factors, and nonseparable pair structure; for higher components, Corollary~\ref{cor:estimated-deflation} adds earlier-axis error. Fixed positive-definite marginal metrics are absorbed by whitening, whereas kernel-edge standardization defines a different correlation estimand; when whitening metrics are estimated, their error adds a matrix perturbation whose effect depends on marginal conditioning.

The results here concern a bilinear objective with fixed pair matrices, including formulations with one positive-definite metric per block. General sums of slide-specific, direction-dependent ratios require separate optimization and perturbation analyses.

Using the oracle bounds as calibrated diagnostics requires control of covariance-estimation error. Randomization calibration of the fitted procedure requires specifying the transformation group, its invariance under the null, and the fitting and selection steps repeated under each transformation. The fourth-corner literature documents, for a different null and permutation scheme, how an unmatched scheme fails to control the type I error~\citep{DrayLegendre2008,terBraak2012}. The present results do not establish the validity of spatial-shift calibration, which translates the positions of one type, for fitted skrCCA axes.

\section*{Use of language models}
The authors used Claude Code (Fable 5.1) and Codex (GPT-6 Astra) to assist with selected mathematical derivations and proof development. The authors independently checked all arguments developed with this assistance and take full responsibility for the mathematical results and all other content of this manuscript.

\end{document}